\documentclass{llncs}

\newif\ifextended
\extendedtrue

\usepackage{llncsdoc}

\usepackage{color}
\usepackage{tikz}
\usetikzlibrary[automata,positioning]

\usepackage[all]{xy}
\usepackage{comment}
\usepackage{cite}
\usepackage{amsmath}
\usepackage{amssymb}
\usepackage{flushend} 
\usepackage{subfigure}
\usepackage[all]{xy}
\usepackage{graphicx}
\usepackage{listings}
\usepackage{amsfonts}
\usepackage{mathrsfs}
\usepackage{eufrak}
\usepackage{todonotes}
\usepackage{oldgerm}

\newcommand{\et}{\wedge}
\newcommand{\biget}{\bigwedge}

\newcommand{\then}{\Rightarrow}

\newcommand{\weakctl}{\mathcal{W}}
\newcommand{\untilctl}{\mathcal{U}}

\newcommand{\glob}{\mathcal{G}}

\newcommand{\future}{\mathcal{F}}

\newcommand{\nat}{\mathbb{N}}
\newcommand{\rat}{\mathbb{Q}}
\newcommand{\real}{\mathbb{R}}

\author{Luca Spalazzi \and Francesco Spegni \\
\institute{DII - Universit\`a Politecnica delle Marche \\
Ancona, Italy}
\email{\{spalazzi,spegni\}@dii.univpm.it}
}

\ifextended
\title{Parameterized Model-Checking of Timed Systems with Conjunctive Guards \\ (Extended Version)}
\else
\title{Parameterized Model-Checking of Timed Systems with Conjunctive Guards}
\fi

\begin{document}

\maketitle

\begin{abstract}
In this work we extend the Emerson and Kahlon's cutoff theorems for process 
skeletons with conjunctive guards to Parameterized Networks of Timed Automata, 
i.e. systems obtained by an \emph{apriori} unknown number of Timed Automata instantiated 
from a finite set $U_1, \dots, U_n$ of Timed Automata templates. In this way we
aim at giving a tool to universally verify software systems where an unknown 
number of software components (i.e. processes) interact with continuous time 
temporal constraints. It is often the case, indeed, that distributed algorithms 
show an heterogeneous nature, combining dynamic aspects with real-time aspects.
In the paper we will also show how to model check a protocol that uses special
variables storing identifiers of the participating processes (i.e. PIDs) in
Timed Automata with conjunctive guards. This is non-trivial, since solutions to
the parameterized verification problem often relies on the processes to be
symmetric, i.e. indistinguishable. On the other side, many popular distributed 
algorithms make use of PIDs and thus cannot directly apply those solutions.
\end{abstract}

\section{Introduction}

Software model-checking emerged as a natural evolution of applying model checking
to verify hardware systems. 
Some factors, among several ones, that still make software model checking 
challenging are: the inherently dynamic nature of software components, the 
heterogeneous nature of software systems and the relatively limited amount of 
modular tools (both theoretical and practical) for verifying generic software 
systems. 

Software systems definable as an arbitrary number of identical copies of some
process template, are called parameterized systems, and are an example of 
infinite state 
systems \cite{Emerson1998}. Sometimes the nature of a software system is 
heterogeneous, meaning that it combines several ``characteristics'' (e.g. a clock
synchronization algorithm is supposed to work with an arbitrary number of processes
but also to terminate within a certain time). The
scarcity of modular tools is witnessed by the fact that almost everyone trying to
model check a software system, has to build his/her own toolchain that applies
several intermediate steps (usually translations and abstractions) before building
a model that can be actually model checked.

Despite such obstacles, several industries already apply model checking as part
of their software design and/or software testing stages. (e.g., 
Microsoft \cite{ball2011}, NASA \cite{Mansouri2008}, 
Bell Labs.\cite{godefroid2005}, IBM \cite{ben2003},
UP4ALL\footnote{http://www.uppaal.com/index.php?sida=203\&rubrik=92 URL visited on April '14}). In the aerospace industry, the DO178C international standard
\cite{RTCA11a}
even consider software model checking (or more generally, software 
verification) an alternative to software testing, under suitable assumptions.

The core of our work is an extension of the Emerson and Kahlon's Cutoff 
Theorem \cite{Emerson2000} to \emph{parameterized and timed systems}.
Assuming a parameterized system based on Timed Automata $U_1, \dots, U_m$ that
synchronize using conjunctive Boolean guards,
the cutoff theorem allows to compute a list of positive numbers 
$(c_1, \dots, c_m)$ such that, let $\phi$ be a given specification,
then:
\[
\begin{array}{l}
\forall i \in [1,m] . (\forall n_i \in [0,\infty) ~ . ~
    (U_1,\dots,U_m)^{(n_1, \dots, n_m)} \models \phi ~ \textit{iff} \\
\hspace{2cm} \forall n_i \in [0,c_i] ~ . ~ 
    (U_1,\dots,U_m)^{(n_1, \dots, n_m)} \models \phi)
\end{array}
\]  
Intuitively, the proof shows that the cutoff configuration is 
\emph{trace equivalent} to each ``bigger'' system.

The contribution of this work is multifold, w.r.t. the aforementioned factors: it
reduces the problem of model checking an \emph{infinite state} real-time software 
system to model checking a finite number of finite state systems; it shows a 
concrete example of how to combine verification algorithms from distinct domains, 
to verify what we call a \emph{heterogeneous} software systems; the cutoff 
theorem for real-time systems is a theoretical tool that can be applied as a 
first step when verifying a parameterized and real-time algorithm. 
A second contribution is methodological: this paper describes how 
to exploit the cutoff theorem to model variables that store process identifiers 
(PIDs) of processes participating to the distributed algorithm. 
This is non trivial, since the former relies on the fact that processes should be symmetrical, thus indistinguishable. 
In order to show this, we will use a popular benchmark protocol, viz. the 
Fischer's protocol for mutual exclusion. To the best of our knowledge, this is 
the first time that the Fischer's protocol has been verified using model checking 
techniques, for an \emph{apriori} unknown number of processes.

\section{Related Work}
\label{sec:rw}

\textbf{Infinite State System.}
Timed Automata and Parameterized Systems are two examples of infinite 
state systems \cite{Emerson1998}. In general, the problem of
model checking infinite state systems is undecidable \cite{Krzysztof1986}. 
A classic approach to overcome this limitation, is to find
suitable subsets of infinite state systems that can be reduced to model checking
of finitely many finite state systems, e.g. identifying a precise abstraction
(e.g. clock-zones for Timed Automata \cite{Bengtsson2004}).
Other approaches are based on the idea of finding a finite-state abstraction 
that is correct but not complete, such that a property verified for the abstract 
system holds for the original system as well 
\cite{Clarke1986,Zuck04a,German92a,Spegni2014}. 
Some other approaches are based on the idea of building an invariant representing 
the common behaviors exhibited by the system \cite{Kurshan89a}. 
When a given relation over the invariant is satisfied, then the desired property 
is satisfied by the original system.
Its limitation is that building the abstraction or the invariant is usually not 
automatic. \\
\textbf{Cutoffs for Parameterized Systems.}
Concerning the use of cutoff for model checking parameterized systems, there 
exists two main approaches: computing the cutoff number of process replications 
or the cutoff length of paths.
The former consists in finding a finite number of process instances such that 
if they satisfy a property then the same property is satisfied by an arbitrary 
number of such processes.
Emerson and Kahlon \cite{Emerson2000} established a cutoff value of about the 
number of template states, for a clique of interconnected process 
skeletons. In the case of rings, a constant between 2 and 5 is enough
\cite{Emerson03}. For shared resources management algorithms 
\cite{Bouajjani08a}, the cutoff value is the number of resources plus the 
quantified processes (in the decidable fragment of processes with equal priority).
Other works proved that one process per template is enough, for certain grids 
\cite{Pagliarecci2011}. 
Recently, in \cite{Rubin2014a} it has been showed that certain parameterized 
systems may admit a cutoff which is not computable, while Hanna \emph{et al.} 
\cite{Hanna10} proposed a procedure to compute a cutoff for Input-Output Automata 
that is independent of the communication topology.
On the other hand, computing the cutoff length of paths 
of a parameterized system consists in finding an upper bound on the number 
of nodes in its longest computation path. 
When a property is satisfied within the bounded path, then the property holds for 
a system with unbound paths, i.e., with an arbitrary number of process instances.
The classic work from German and Sistla \cite{German92a}, Emerson and 
Namjoshi \cite{Emerson96a} proved that such a cutoff exists for the verification of
parameterized systems composed of a control process and
an arbitrary number of user processes against indexed \textsc{ltl}
properties.
Yang and Li \cite{Yang10a} proposed a sound and complete method to compute such 
a cutoff for parameterized systems with only rendezvous actions. 
In that work, the property itself is represented as an automaton.
Lately it has been also showed that parameterized systems on pairwise rendezvous
do not admit, in general, a cutoff \cite{Spegni2014}.
To the best of our knowledge, cutoff theorems have not been stated previously for 
timed systems. Surprisingly enough, extending Emerson and Kahlon cutoff theorems
\cite{Emerson2000} to timed systems does not increase the cutoff value.\\
\textbf{Parameterized Networks of Timed or Hybrid Automata.}
The realm of real-time systems (timed automata and, more in general, hybrid 
automata) with a finite but unknown number of instances has been explored.
Abdulla and Jonsson \cite{Jonsson} proposed in their seminal work to 
reduce safety properties to reachability properties. 
They worked with a network composed by an arbitrary set of identical timed 
automata controlled by a controller (i.e. a finite timed automaton as well).
Abdulla et al. show also that checking safety properties in networks of timed 
automata
with multiple clocks is an undecidable problem \cite{Abdulla04}, as well as the
problem of determining if a state is visited infinitely often, in the continuous 
time model (in the discrete time model, instead, it is decidable) 
\cite{Abdulla2003}.
It should be remarked that in their undecidability proof, the network of timed
automata must rely on synchronous rendezvous in order to prove the undecidability
results. This motivated us to explore timed automata with different synchronization
mechanisms in this work.
Ghilardi et al. \cite{Carioni2010}, reduced model checking
safety properties to reachability problem.
Similarly to Abdulla and Jonsson, they applied their approach to networks 
composed by an arbitrary set of timed automata interacting with a controller.
Their original contribution consisted in the usage of
Satisfiability Modulo Theories techniques.
G\"{o}thel and Glesner \cite{Gothel10a} proposed a semi-automatic verification 
methodology based on finding network invariants and using both theorem proving 
and model checking.
Along the same line, Johnson and Mitra \cite{johnson12} proposed a 
semi-automatic verification of safety properties for parameterized networks of 
hybrid automata with rectangular dynamics.
They based their approach on a combination of invariant synthesis and 
inductive invariant proving. Their main limitation is that specifications 
are often not inductive properties (e.g. the
mutual exclusion property it is not an inductive property). In this case one
must show that a set of inductive invariants can imply the desired property. 
This last step is often not fully automatic.

We consider systems composed of a finite number of templates, each of which can 
be instantiated an arbitrary number of times.
We limit Timed Automata to synchronize using Conjunctive Guards, instead of the 
classic Pairwise Rendezvous \cite{Bengtsson2004}, 
because, as already mentioned, parameterized systems with pairwise rendezvous do 
not admit, in general, a cutoff \cite{Spegni2014}.
Finally, the verification proposed in this paper is completely automatic.

\section{Parameterized Networks of Timed Automata}
\label{sec:theoretical-aspects}
\label{ssec:pnta}

This work introduces Parameterized Networks of Timed Automata (PNTA),
an extensions of Timed Automata that synchronize using conjunctive Boolean
guards.
We also introduce Indexed-Timed CTL$^\star$, a temporal logic that integrates 
TCTL and MTL \cite{Bouyer2009}, for reasoning about timed processes,
together with Indexed-CTL$^\star\setminus$X \cite{Emerson2000}, for reasoning 
about parametric networks of processes. In the following definition we will 
make use of a set of \emph{temporal constraints} $TC(C_l)$, defined as:
\[
\begin{array}{rclclclclcl}
\textit{TC}(C) &::=& \top ~ | ~ \neg~ TC(C) ~ | ~ TC(C) ~\vee~ TC(C) ~ | \\
                 && C ~ \sim ~ C ~ | ~ C ~\sim~ \rat^{\geq 0} \\ 

\end{array}
\]
where $\sim ~ \in ~ \{ <, \leq, >, \geq, = \}$, $C$ is a set of clock variables and $\rat$ denotes the set of 
rational numbers.

\begin{definition}[\bf Timed Automaton Template]
\label{def:ta}
A Timed Automaton (TA) Template $U_l$ is a tuple 
$\langle S_l, \hat{s}_l, C_l, \Gamma_l, \tau_l, I_l\rangle$ where:
\begin{itemize}
  \item $S_l$ is a finite set of states, or locations;
  \item $\hat{s}_l \in S_l$ is a distinguished initial state;
  \item $C_l$ is a finite set of clock variables;
  \item $\Gamma_l$ is a finite set of Boolean guards built upon $S_l$;
  \item $\tau_l \subseteq S_l \times TC(C_l) \times 2^{C_l} \times \Gamma_l \times S_l$ 
        is a finite set of transitions;
  \item $I_l : S_l \rightarrow TC(C_l)$ maps a state to an invariant, such that $I_l(\hat{s}_l)=\top$;
\end{itemize}
\end{definition}
We will denote with $|U_l| = |S_l|$ the size of the timed automaton. A network of 
timed automata can be defined as a set of $k$ TA templates, where each TA
template (say $U_l$) is instantiated an arbitrary number (say $n_l$) of times. 
\begin{definition}[\bf PNTA]
\label{def:pnta}
Let $(U_1, \dots, U_k)$ be a set of Timed Automaton templates.
Let $(n_1, \dots, n_k)$ be a set of natural numbers.
Then
$$
(U_1, \dots, U_k)^{(n_1, \dots, n_k)}
$$
is a Parameterized Network of Timed Automata denoting the asynchronous parallel
composition of timed automata $U_1^1 || \dots || U_1^{n_1} || \dots || U_k^1 || \dots || U_k^{n_k}$, such that for each $l \in [1,k]$ and $i \in [1,n_l]$, then
$U_l^i$ is the \emph{i-th copy} of $U_l$.
\end{definition}
Let us remark that every component of $U_l^i$ is a disjoint copy of the 
corresponding template component.
In the following will be described how every process $U_l^i$, also called 
instance, can take a local step after having checked that the neighbors' states 
satisfy the transition (conjunctive) Boolean guard. In such system a process can 
check it is ``safe'' to take a local step, but it cannot induce a move on a 
different instance.
A PNTA based on conjunctive guards is defined as follows.
\begin{definition}[\bf PNTA with Conjunctive Guards]
\label{def:cg-pnta}
Let $(U_1, \dots, U_k)^{(n_1, \dots, n_k)}$ be a PNTA.
Then, it is a PNTA with Conjunctive Guards iff every $\gamma \in \Gamma_l^i$ is a
Boolean expression with the following form:
\[
	\bigwedge_{\substack{m \in [1,n_1] \\ m \neq i}} (\hat{s}_l(m) \vee s_l^1(m) \vee \dots \vee s_l^p(m)) ~ \wedge ~ 
	\bigwedge_{\substack{h \in [1,k] \\ h \neq l}}(\bigwedge_{j \in [1,n_j]} (\hat{s}_h(j) \vee s_h^1(j) \vee \dots \vee s_h^q(j)))
\]
where, for all $l \in [1,k]$, $i \in [1,n_l]$ and $p > 0$, 
$\{ s_l^1, \dots, s_l^p \} \subseteq S_l$, $s_l(i) \in S_l^i$ and 
$\hat{s}_l$ is the initial states of $U_l$. The initial states $\hat{s}_l(m)$ and
$\hat{s}_h(j)$ must be present.
\end{definition}
We remark that our definitions of Timed Automaton template, PNTA and PNTA with 
Conjunctive Guards are variants of the notion of \emph{timed automata} and 
\emph{networks of timed automata} found in literature (e.g. \cite{Bengtsson2004}).

The \emph{operational semantics} of PNTA with conjunctive guards is expressed as 
a transition system over \emph{PNTA configurations}.
\begin{definition}[\bf PNTA Configuration]
\label{def:abs-conf}
~~ \\
Let $(U_1, \dots, U_k)^{(n_1, \dots, n_k)}$ be a PNTA.
Then a \emph{configuration} is a tuple:
$$
\mathfrak{c} = (\langle \overline{s}_1, \overline{u}_1 \rangle, \dots, \langle \overline{s}_k, \overline{u}_k \rangle)
$$
where, for each $l \in [1,k]$:
\begin{itemize}
\item $\overline{s}_l : [1,n_l] \to S_l$ 
    maps an instance to its current state, and
\item $\overline{u}_l : [1,n_l] \to (C_l \to \real^{\geq 0})$, 
    maps an instance to its clock function, s.t.
	\begin{equation}
	\label{eq:invariant}
	   \forall i ~ . ~ \overline{u}_l(i) \models I_l^i(\overline{s}_l(i))
	\end{equation}
\end{itemize}
$\mathfrak{C}$ is the set of all the configurations.
\end{definition}
Intuitively, let $(\dots, \langle \overline{s}_l, \overline{u}_l \rangle, \dots)$ 
be a configuration, then $\overline{s}_l(i)\in S_l$ denotes the state where 
instance $U_l^i$ is in that configuration. 
$\overline{u}_l(i)$ is the clock assignment function (i.e., 
$\overline{u}_l(i): C_l \to \real^{\geq 0}$) of instance $U_l^i$ in that 
configuration. 
In other words, for each $c\in C_l$, $\overline{u}_l(i)(c)$ is the current value 
that the clock variable $c$ assumes for instance $U_l^i$.
Any assignment to such clock variables must satisfy the invariant for the 
corresponding state (see Eqn. (\ref{eq:invariant})).
The notion of transition requires some auxiliary notations.
Let $l\in[1,k]$, and let $i\in[1,n_l]$, then we call:
\begin{itemize}
  \item \emph{initial configuration} \\
        $\hat{\mathfrak{c}} \in \mathfrak{C}$ such that,
        for each $l\in[1,k]$, for each $i\in[1,n_k]$: \\
        \hspace*{1.3cm} $\overline{s}_l(i)=\hat{s}_l^i$, and \\
        \hspace*{1.3cm} $\forall c \in C_l$, $\overline{u}_l(i)(c) = 0$.
  \item \emph{projection} \\
        $\forall \mathfrak{c} = 
                 (\langle \overline{s}_1, \overline{u}_1 \rangle, \dots, 
                  \langle \overline{s}_l, \overline{u}_l \rangle, \dots, 
	              \langle \overline{s}_k, \overline{u}_k \rangle) \in \mathfrak{C}$, \\
        \hspace*{1.3cm} $\mathfrak{c}(l)   = \langle \overline{s}_l,    \overline{u}_l    \rangle$, and \\
        \hspace*{1.3cm} $\mathfrak{c}(l,i) = \langle \overline{s}_l(i), \overline{u}_l(i) \rangle$.
  \item \emph{state-component} \\ 
        $\forall \mathfrak{c} = 
                 (\langle \overline{s}_1, \overline{u}_1 \rangle, \dots, 
                  \langle \overline{s}_l, \overline{u}_l \rangle, \dots, 
	              \langle \overline{s}_k, \overline{u}_k \rangle) \in \mathfrak{C}$, \\
        \hspace*{1.3cm} $\textit{state}(\mathfrak{c})      = (\overline{s}_1, \dots,\overline{s}_l, \dots,\overline{s}_k)$, \\
        \hspace*{1.3cm} $\textit{state}(\mathfrak{c}(l))   = \overline{s}_l$, and \\
        \hspace*{1.3cm} $\textit{state}(\mathfrak{c}(l,i)) = \overline{s}_l(i)$.
  \item \emph{clock-component} \\ 
        $\forall \mathfrak{c} = 
                 (\langle \overline{s}_1, \overline{u}_1 \rangle, \dots, 
                  \langle \overline{s}_l, \overline{u}_l \rangle, \dots, 
	              \langle \overline{s}_k, \overline{u}_k \rangle) \in \mathfrak{C}$, 
	    $\forall c \in C_l$,\\
        \hspace*{1.3cm} $\textit{clock}(\mathfrak{c})         = (\overline{u}_1, \dots,\overline{u}_l, \dots,\overline{u}_k)$, \\
        \hspace*{1.3cm} $\textit{clock}(\mathfrak{c}(l))      = \overline{u}_l$, \\
        \hspace*{1.3cm} $\textit{clock}(\mathfrak{c}(l,i))    = \overline{u}_l(i)$, thus \\
        \hspace*{1.3cm} $\textit{clock}(\mathfrak{c}(l,i))(c) = \overline{u}_l(i)(c)$. 
  \item \emph{time increase} \\
  		$\forall c \in C_l . \forall d \in \real^{\geq 0} . 
		         (\overline{u}_l + d)(i)(c) = \overline{u}_l(i)(c) + d$ \\
        \hspace*{1.3cm} $(\textit{clock}(\mathfrak{c})+d)      = (\overline{u}_1+d, \dots,\overline{u}_l+d, \dots,\overline{u}_k+d)$, \\
        \hspace*{1.3cm} $(\textit{clock}(\mathfrak{c}(l))+d)   = (\overline{u}_l+d)$, and \\
        \hspace*{1.3cm} $(\textit{clock}(\mathfrak{c}(l,i))+d) = (\overline{u}_l+d)(i)$.
  \item \emph{clock reset} \\
	    $\forall c \in C_l ~ . ~ \forall r \subseteq C_l ~ . ~ \forall j ~ . $ \\
	    \hspace*{1.3cm}
 	    $\overline{u}_l[(i,r) \mapsto 0](j)(c) = \left\{
	        \begin{array}{ll}
	  	      0                    & \textrm{if } i = j ~ \textit{and} ~ c \in r \\
	  	      \overline{u}_l(j)(c) & \textrm{otherwise}
	        \end{array}
	     \right.$ 
       
  \item \emph{clock constraint evaluation} \\
        $\overline{u}_l(i) \models g$ \emph{iff} 
        the clock values of instance $U_l^i$ denoted by $\overline{u}_l(i)$
        satisfy the clock constraint $g$; the semantics $\models$ is defined as 
        usual by induction on the structure of $g$;
  \item \emph{guard evaluation} \\
	    $\textit{state}(\mathfrak{c}) \models \gamma$ \emph{iff} 
	    the set of states of $(U_1, \dots, U_k)^{(n_1, \dots, n_k)}$ 
	    denoted by $\textit{state}(\mathfrak{c})$
	    satisfies the Boolean guard $\gamma$; this predicate as well can be
        defined by induction on the structure of $\gamma$.
\end{itemize}

\begin{definition}[\bf PNTA Transitions]
\label{def:sem-op}
The transitions among PNTA configurations are governed by the following rules:
\vspace{0.25cm}
\\
$
\textit{(delay)} \\
\hspace*{1.0cm}   \mathfrak{c} \xrightarrow{d} \mathfrak{c}^\prime \hspace{1cm} 
\textrm{if} \hspace{0.3cm} d \in \real^{\geq 0} \\
\hspace*{3.5cm}     \textit{state}(\mathfrak{c}^\prime) = \textit{state}(\mathfrak{c}) \\
\hspace*{3.5cm}     \textit{clock}(\mathfrak{c}^\prime) = (\textit{clock}(\mathfrak{c})+d) \\
\hspace*{3.5cm}     \forall l, i, d^\prime \in [0,d] . \textit{clock}(\mathfrak{c}(l,i))+d^\prime \models I_l^i(\textit{state}(\mathfrak{c}(l,i)))
\vspace{0.25cm}
\\
\textit{(synchronization)} \\
\hspace*{1.0cm}   \mathfrak{c} \xrightarrow{\gamma} \mathfrak{c}^\prime \hspace{0.5cm} \textrm{if} \hspace{0.3cm} \exists l\in[1,k] ~ . \exists i\in[1,n_l] ~ :\\
\hspace*{3.5cm}       s \xrightarrow{g,r,\gamma} t \in \tau_l^i. \\
\hspace*{3.5cm}       \textit{state}(\mathfrak{c}(l,i)) = s, \\
\hspace*{3.5cm}       \textit{clock}(\mathfrak{c}(l,i)) \models g, \\
\hspace*{3.5cm}       \textit{state}(\mathfrak{c}) \models \gamma, \\
\hspace*{3.5cm}       \mathfrak{c}^\prime(h) = \mathfrak{c}(h) ~ \textrm{for each} ~ h \neq l, \\
\hspace*{3.5cm}       \mathfrak{c}^\prime(l,j) = \mathfrak{c}(l,j) ~ \textrm{for each} ~ j \neq i, \\
\hspace*{3.5cm}       \mathfrak{c}^\prime(l,i) = \langle t , \textit{clock}(\mathfrak{c}(l,i))[ (r,i) \mapsto 0] \rangle \\
\hspace*{3.5cm}       \textit{clock}(\mathfrak{c}^\prime(l,i)) \models I_l^i(\textit{state}(\mathfrak{c}^\prime(l,i)) \\
$
\end{definition}
%
Let us define what is a \emph{timed-computation} for PNTA.
\begin{definition}[\bf Timed Computation]
\label{def:timed-computation}
Let $\hat{\mathfrak{c}_0}$ be an \emph{initial configuration}, a 
\emph{timed-computation} $x$ is a finite or infinite sequence of pairs:
$$
x = (\mathfrak{c}_0, t_0) \dots (\mathfrak{c}_v, t_v) \dots  
$$
s.t. $t_0 = 0$ and 
$\forall v \geq 0 ~ . ~ (\exists d > 0 ~ . ~ \mathfrak{c}_v \xrightarrow{d} \mathfrak{c}_{v+1} ~ \wedge ~ t_{v+1} = t_v + d) ~ \vee ~ (\exists \gamma ~ . ~ \mathfrak{c}_v \xrightarrow{\gamma} \mathfrak{c}_{v+1} ~ \wedge ~ t_{v+1} = t_v)$
\end{definition}

In other words, a timed computation can be seen as a sequence of \emph{snapshots} 
of the transition system configurations taken at successive times.
It should be noticed that, according to Emerson and Kahlon \cite{Emerson2000}, 
in this work, it has been adopted the so-called \emph{interleaving semantics}.
This means that in a transition between two configurations, only one instance 
can change its state (see the \emph{synchronization} rule in 
Def. \ref{def:sem-op}).
For the sake of conciseness, let us extend the notion of 
\emph{projection}, \emph{state-component}, and \emph{clock-component} to timed computations.
Let $x = (\mathfrak{c}_0, t_0) \dots (\mathfrak{c}_v, t_v) \dots$ be a timed computation, let $x_v = ( \mathfrak{c}_v   , t_v )$ be the $v$-th element of $x$, then
\[
\begin{array}{lllclll}

x(l)                     & = & ( \mathfrak{c}_0(l)   , t_0 ) \dots ( \mathfrak{c}_v(l)   , t_v ) \dots &&   
    \textit{clock}(x_v)      & = & \textit{clock}( \mathfrak{c}_v      ) \\

x(l,i)                   & = & ( \mathfrak{c}_0(l,i) , t_0 ) \dots ( \mathfrak{c}_v(l,i) , t_v ) \dots &&  
    \textit{clock}(x_v(l))   & = & \textit{clock}( \mathfrak{c}_v(l)   ) \\

x_v(l)                   & = & ( \mathfrak{c}_v(l)   , t_v ) && 
    \textit{clock}(x_v(l,i)) & = & \textit{clock}( \mathfrak{c}_v(l,i) ) \\
x_v(l,i)                 & = & ( \mathfrak{c}_v(l,i) , t_v ) && \\
\textit{state}(x_v)      & = & \textit{state}( \mathfrak{c}_v      ) &&
      \textit{time}(x_v)       & = & t_v  \\
\textit{state}(x_v(l))   & = & \textit{state}( \mathfrak{c}_v(l)   ) &&
     \textit{time}(x_v(l))    & = & t_v  \\ 
\textit{state}(x_v(l,i)) & = & \textit{state}( \mathfrak{c}_v(l,i) ) &&
     \textit{time}(x_v(l,i))  & = & t_v  \\
\\

\end{array}
\]
$x(l,i)$ is called the \emph{local computation} of the $i$-th instance of automaton template $l$.
$\textit{time}(x_v)$, $\textit{time}(x_v(l))$, and $\textit{time}(x_v(l,i))$ 
are the \emph{time-components} of $x_v$, $x_v(l)$, and $x_v(l,i)$ respectively.  

\begin{definition}[\bf Idle Local Computation]
\label{def:idle-computation}
Let $U_l^i = \langle S_l^i, \hat{s}_l^i, C_l^i, \tau_l^i, I_l^i \rangle$ be the $i$-th instance of the \emph{timed automaton template} $U_l$.
An \emph{idle local computation} $\hat{\mathfrak{s}}(l,i)$ is a timed local computation such that, for all $v \geq 0$:
\begin{eqnarray*}
\hat{\mathfrak{s}}(l,i) & = &
( \langle\hat{s}_l^i,\overline{u}_l(i)\rangle , t_0) \dots ( \langle\hat{s}_l^i,\overline{u}_l(i)+t_v\rangle , t_v) \dots   \\
\hat{\mathfrak{s}}_v(l,i) & = & ( \langle\hat{s}_l^i,\overline{u}_l(i)+t_v\rangle , t_v)
\end{eqnarray*}
where $t_0=0$ and for each $c \in C_l$, $\overline{u}_l(i)(c) = 0$.
\end{definition}
It should be noticed that for each $v$, it must be 
$\overline{u}_l(i)+t_v\models I_l^i(\hat{s}_l^i)$, 
since $I_l^i(\hat{s}_l^i)=\top$ according to Def. \ref{def:ta}.
Intuitively, an idle local computation is an instance of the automaton template $U_l$ that stutters in its initial state.
\begin{definition}[\bf Stuttering]
\label{def:stuttering}
Let $x$ and $y$ be two timed computations.
Let $x = x_0 \cdot \ldots \cdot x_v \cdot x_{v+1} \ldots$ 
The timed computation $y$ is a \emph{stuttering} of the timed computation $x$ iff 
for all $v \geq 0$, there exists $r \geq 0$, such that
\begin{eqnarray*}
y & = & x_0 \cdot \ldots \cdot x_v \cdot x_{v,\delta_1} \cdot x_{v,\delta_2} \cdot \ldots \cdot x_{v,\delta_r} \cdot x_{v+1} \ldots
\end{eqnarray*}
where $\delta_1, \delta_2, \dots, \delta_r \in \real^{\geq 0}$, 
$\delta_1 \leq \delta_2 \leq \dots \leq \delta_r$, 
$t_v+\delta_r \leq t_{v+1}$, and \\
\hspace*{1cm}$x_{v,\delta_1} = (\langle\textit{state}(x_v) , \textit{clock}(x_v)+\delta_1\rangle , t_v+\delta_1)$ \\
\hspace*{1cm}$x_{v,\delta_2} = (\langle\textit{state}(x_v) , \textit{clock}(x_v)+\delta_2\rangle , t_v+\delta_2)$ \\
\hspace*{1cm}$\dots$ \\
\hspace*{1cm}$x_{v,\delta_r} = (\langle\textit{state}(x_v) , \textit{clock}(x_v)+\delta_r\rangle , t_v+\delta_r)$ \\
\end{definition}
Intuitively, the above definition means that a stuttering of a given timed 
computation $x$ can be generated by inserting an arbitrary number of 
\emph{delay transitions} (see Def. \ref{def:sem-op}) short enough to not alter 
the validity of temporal conditions of the original computation $x$.
It only represents a more detailed view (i.e. a finer sampling) of the interval between 
a configuration and the next one without changing the original sequence of states.

For the purpose of this work, timed computations conforming to Def. \ref{def:timed-computation}  
(i.e. each configuration complies with Eqn. (\ref{eq:invariant})) can be classified in three different kinds 
of computation:
\begin{itemize}
\item \emph{Infinite Timed Computation}:
      $x$ is a timed computation of infinite length.
\item \emph{Deadlocked Timed Computation}:
      $x$ is a maximal finite timed computation, i.e. in it reaches a final
      configuration where all transitions are disabled.
\item \emph{Finite Timed Computation}:
      $x$ is a (not necessarily maximal) final timed computation, i.e. it is 
      either a deadlocked computation or a finite prefix of an infinite one.
\end{itemize}
%

\section{A Temporal Logic for PNTA}
\label{ssec:iltl}

A dedicated logic is needed in order to specify behaviors of a PNTA. 
This logic, named Indexed-Timed-CTL$^\star$, allows to reason about real-time intervals
and temporal relations (until, before, after, \dots) in systems of arbitrary size.
While its satisfiability problem is undecidable, the problem of model checking a 
PNTA is proved to be decidable, under certain conditions.

\begin{definition}[\bf Indexed-Timed-CTL$^\star$]
\label{def:ITCTL}
Let $\{ P_l \}_{l \in [1,k]}$ be finite sets of atomic propositions.
Let $p(l,i)$ be any atomic proposition such that $l \in [1,k]$, $i\in\nat^{>0}$, 
and $p\in P_l$.
Then, the set of ITCTL$^\star$ formulae is inductively defined as follows:
$$
\begin{array}{rcl}
\phi &::=& \top ~ | ~ p(l,i) ~ | ~ \phi ~ \et ~ \phi ~ | ~ \neg \phi ~ | ~ \bigwedge_{i_l} \phi ~ | ~ A \Phi ~ | ~ A_\textit{fin} \Phi ~ | ~ A_\textit{inf} \Phi \\
\Phi &::=& \phi ~|~ \Phi ~ \et ~ \Phi ~|~ \neg \Phi ~|~ \Phi ~ \untilctl_{\sim q} ~ \Phi \\
\end{array}
$$
where $\sim ~ \in \{ <, \leq, \geq, >, = \}$ and $q \in \rat^{\geq 0}$.
\end{definition}

As usual for branching-time temporal logics, the terms in $\phi$ denote
\emph{state} formulae, while terms in $\Phi$ denote \emph{path} formulae. 
%
For the purpose of this work it is enough to assume the set of atomic propositions 
coincides with the set of states of a given PNTA, i.e. $P_l = S_l$, for every $l$.

The path quantifier $A_\textit{fin}$ (resp. $A_\textit{inf}$) is a variant of the
usual universal path quantifier $A$, restricted to paths that are of finite
length (resp. infinite length). Such variants are inspired by \cite{Emerson2000}.
Missing Boolean ($\lor,\to, \dots$) operators, 
temporal operators ($\glob, \future, \weakctl, \dots$), 
as well as path quantifiers ($E, E_\textit{fin}, E_\textit{inf}$) 
can be defined as usual.
%
The semantics of ITCTL$^\star$ is defined w.r.t. a Kripke Structure
integrating the notions of parametric system size and continuous time semantics 
\cite{Bouyer2009}.
The continuous time model requires that between any two configurations it 
always exists a third state. It is possible, though, introduce \emph{continuous 
time computation trees}
\cite{Alur1990}. Let us call \emph{s-path} a function
$\rho : \real^{\geq 0} \to \mathfrak{C}$ that intuitively maps
a time $t$ with the current system configuration at that time. The mapping
$\rho_{\rfloor_{t^\prime}}:[0,t^\prime) \to \mathfrak{C}$ is a 
\emph{prefix} of $\rho$ iff 
$\forall t < t^\prime . 
    \rho_{\rfloor_{t^\prime}}(t)=\rho(t)$. The mapping
$\rho_{\lfloor_{t^\prime}}:[t^\prime,\infty) \to \mathfrak{C}$ 
is a \emph{suffix} of $\rho$ iff 
$\forall t \geq t^\prime . 
    \rho_{\lfloor_{t^\prime}}(t)=\rho(t)$.
Let us take a prefix $\rho_{\rfloor_{t^\prime}}$ and an s-path
$\rho^\prime$, then their \emph{concatenation} is defined as:
\[
	  (\rho_{\rfloor_{t^\prime}} \cdot \rho^\prime)(t) = \left\{
	        \begin{array}{ll}
	  	      \rho_{\rfloor_{t^\prime}}(t) & ~ \textrm{if } t < t^\prime \\
	  	      \rho^\prime(t - t^\prime)                   & ~ \textrm{else}
	        \end{array}
	     \right. \\
\]
Let $\Pi$ be a set of s-paths, then 
$\rho_{\rfloor_{t^\prime}} \cdot \Pi = 
	   \{\rho_{\rfloor_{t^\prime}} \cdot \rho^\prime : \rho^\prime \in \Pi\}$.
A \emph{continuous time computation tree} is a mapping 
$f : \mathfrak{C} \to 2^{[ \real^{\geq 0} \to \mathfrak{C} ]}$ such that:
\[
\begin{array}{l}
\forall \mathfrak{c} \in \mathfrak{C}. 
    \forall \rho \in f(\mathfrak{c}). 
    \forall t \in \real^{\geq 0} ~ . ~
    \rho_{\rfloor_t} \cdot f(\rho(t)) 
        \subseteq f(\mathfrak{c}). 
\end{array}
\]
%
%
For the purpose of this work, here only \emph{s-paths} defined over \emph{timed
computations} will be considered.

\begin{definition}[\bf PNTA s-paths] \label{def:spath} \\
For each timed computation 
$x = (\mathfrak{c}_0, t_0) \dots (\mathfrak{c}_v, t_v) \dots,$
let us call \emph{PNTA s-path} the s-path
$\rho : \real^{\geq 0} \to \mathfrak{C}$ satisfying:
\[
\forall v . \forall t \in [t_v, t_{v+1}) ~ . ~ \rho(t) = 
    \langle s,c \rangle
\]
where $s = \textit{state}(\mathfrak{c}_v)$ and $c = \textit{clock}(\mathfrak{c}_v) + t - t_v$.

\end{definition}

%
It should be noticed that, according to the above construction, an infinite set of
timed computations can generate the same s-path $\rho$; let us denote such set
by $\mathit{tcomp}(\rho)$.
As a consequence, for each $y \in \mathit{tcomp}(\rho)$, there exists $x \in \mathit{tcomp}(\rho)$
such that $y$ is a stuttering of $x$ (see Def. \ref{def:stuttering}).
The continuous semantics of ITCTL$^\star$ can be defined as follows.

\begin{definition}[\bf Satisfiability of ITCTL$^\star$]
~~ \\
Let $(U_1, \dots, U_k)^{(n_1, \dots, n_k)}$ be a PNTA and $\mathfrak{c}$ be the
current configuration.
Let $\phi$ denote an ITCTL$^\star$ state formula, then the 
\emph{satisfiability relation} $\mathfrak{c} \models \phi$
is defined by structural induction as follows:
\[
\begin{array}{rlcl}
\mathfrak{c} \models &\top \\
\mathfrak{c} \models &p(l,i)                 & \textit{ iff } & p = \textit{state}(\mathfrak{c}(l,i)) \\
\mathfrak{c} \models &\phi_1 ~ \et ~ \phi_2  & \textit{ iff } & \mathfrak{c} \models \phi_1 ~ \textit{and} ~ \mathfrak{c} \models \phi_2 \\
\mathfrak{c} \models &\neg \phi_1            & \textit{ iff } & \mathfrak{c} \not\models \phi_1 \\
\mathfrak{c} \models &A \phi_1               & \textit{ iff } & \rho \models \phi_1, \hfill 

                                              \textit{for all} ~ \rho \in f(\mathfrak{c}) ~ \textit{and} ~ \\
                & && \hfill (| \rho | = \omega ~ \textit{or} ~ \textit{deadlock}(\rho)) \\ 
\mathfrak{c} \models &A_\textit{inf} \phi_1               & \textit{ iff } & \rho \models \phi_1, \hfill 

                                              \textit{for all} ~ \rho \in f(\mathfrak{c}) ~ \textit{and} ~ |\rho| = \omega \\ 
\mathfrak{c} \models &A_\textit{fin} \phi_1               & \textit{ iff } & \rho \models \phi_1, \hfill 

                                              \textit{for all} ~ \rho \in f(\mathfrak{c}) ~ \textit{and} ~ |\rho| < \omega \\ 
\
\
\mathfrak{c} \models &\biget_{i_l} \phi(i_l) & \textit{ iff } & \mathfrak{c} \models \phi_1(i_l), \hfill \textit{for each } i_l \in [1,n_l] \\
\\
\rho \models &\phi_1            & \textit{ iff } & \rho(0) \models \phi_1 \\
\rho \models &\phi_1~\et~\phi_2 & \textit{ iff } & \rho \models \phi_1 ~ \textit{and} ~ \rho \models \phi_2 \\
\rho \models &\neg \phi_1       & \textit{ iff } & \rho \not\models \phi_1 \\
\rho \models &\phi_1 ~ \untilctl_{\sim q} ~ \phi_2 
                                   & \textit{ iff } & \textit{for some} ~ t^\prime \sim q, \textit{where } \sim ~ \in \{ <, \leq,\geq,>,= \} \\
                                  &&                & \hfill \rho_{\lfloor_{t^\prime}} \models \phi_2, \textit{ and } \rho_{\lfloor_t} \models \phi_1 ~ \textit{for all} ~ t \in [0,t^\prime) \\
\end{array}
\]
where $| \rho| = \omega$ (resp. $|\rho| < \omega$, resp.
$\textit{deadlock}(\rho)$) denotes that the s-path $\rho$
has infinite length (resp. has finite length, resp. is deadlocked).
\end{definition}
Note that a finite s-path is not necessarily deadlocked, since it can be a
finite prefix of some infinite s-path. 
When a given PNTA $(U_1, \dots, U_k)^{(n_1, \dots, n_k)}$ 
satisfies an ITCTL$^\star$ state-formula $\phi$ at its initial configuration $\hat{\mathfrak{c}}$, 
this is denoted by
\[
(U_1, \dots, U_k)^{(n_1, \dots, n_k)} \models \phi
\]

\begin{theorem}[\bf Undecidability of ITCTL$^\star$]
The satisfiability problem for ITCTL$^\star$ is undecidable.
\end{theorem}
\begin{proof}
The satisfiability problem for TCTL is undecidable \cite{Alur1990}.
TCTL is included in ITCTL$^\star$, therefore the latter is undecidable.
\end{proof}
In the next section we will call \textit{IMTL} the fragment of ITCTL$^\star$ 
having formulae with the following forms: 
$\bigwedge_{i_l} Q h(i_l)$, where $Q \in \{ A, A_\textit{fin}, A_\textit{inf} \}$ 
and in $h$ only Boolean ($\wedge$ and $\neg$) and temporal 
($\untilctl_{\sim q}$) operators are allowed. We will call \textit{IMITL} the 
subset of IMTL where equality constraints (i.e. $\untilctl_{= q}$) are excluded.

\section{Cutoff Theorem for PNTA with Conjunctive Guards}
\label{ssec:cutoff}
\label{ssec:conj-cutoff}

In this section we prove that a cutoff can be computed to make the PMCP of PNTAs 
with conjunctive guards decidable, for a suitable set of formulae. 
The system in which every template is instantiated as many times as its cutoff,
will be called the \emph{cutoff system}.
Given two instantiations $I = (U_1, \dots, U_k)^{(c_1, \dots,c_k)}$ and 
$I^\prime = (U_1, \dots, U_k)^{(c_1^\prime, \dots,c_k^\prime)}$, such that all 
$c_i^\prime \geq c_i$ and at least one $c_j^\prime > c_j$, it can be said that 
$I^\prime$ is \emph{bigger} than $I$, written $I^\prime > I$.
The cutoff theorem states that given a cutoff system $I$, for each $I^\prime > I$, 
both $I^\prime$ and $I$ satisfy the same subset of ITCTL$^\star$ formulae.

\begin{theorem}[\bf Conjunctive Cutoff Theorem]
\label{the:conj-cutoff}
\ \ \\
Let $(U_1, \dots, U_k)$ be a set of TA templates with conjunctive guards.
Let $\phi = \bigwedge_{i_{l_1},\dots,i_{l_h}}Q \Phi(i_{l_1},\dots,i_{l_h})$  
where $Q \in \{ A, A_\textit{inf}, A_\textit{fin}, 
E, E_\textit{inf}, E_\textit{fin} \}$ and $\Phi$ is an IMTL formula and 
$\{ l_1, \dots, l_h \} \subseteq [1,k]$. Then
$$
\begin{array}{l}
\forall (n_1,\dots,n_k) . (U_1, \dots, U_k)^{(n_1, \dots,n_k)} \models \phi 
~ ~ \textit{iff } \\
\hspace{1cm} \forall (d_1,\dots,d_k) \preceq (c_1, \dots,c_k) . (U_1, \dots, U_k)^{(d_1, \dots,d_k)} \models \phi
\end{array}
$$
where the cutoff $(c_1,\dots,c_k)$ can be computed as follows:
\begin{itemize}
\item In case $Q \in \{ A_\textit{inf}, E_\textit{inf} \}$ 
      (i.e., deadlocked or finite timed computations are ignored). Then
      $c_{l}=2$ if $l\in\{l_1,\dots,l_h\}$, and  
      $c_{l}=1$ otherwise (i.e. $l \in [1,k] \setminus \{ l_1,\dots,l_h\}$).
\item In case $Q \in \{ A_\textit{fin}, E_\textit{fin} \}$ (i.e. finite timed 
      computations, either deadlocked or finite prefixes of infinite 
      computations). Then
      $c_{l}=1$ for each $l$.
\item In case $Q \in \{ A, E \}$ (i.e., infinite and deadlocked). Then 
      $c_{l}=2|U_{l}|+1$ if $l\in\{l_1,\dots,l_h\}$;
      $c_{l}=2|U_{l}|$ otherwise (i.e. $l \in [1,k] \setminus \{l_1,\dots,l_h\}$).
\end{itemize}
\end{theorem}

The proof of the Cutoff Theorem consists of three steps.
The first step (\emph{Conjunctive Monotonicity Lemma}) shows that 
adding instances to the system does not alter the truth of logic formulae. 
The second step (\emph{Conjunctive Bounding Lemma}) proves that removing an
instance beyond the cutoff number, does not alter the truth of logic formulae
either.
The third step (\emph{Conjunctive Truncation Lemma}) generalizes the 
Conjunctive Bounding Lemma to a system that has two automaton templates with an 
arbitrary number of instances.
The given proofs can be generalized to systems with an arbitrary number of
templates.

\begin{theorem}[\bf Conjunctive Monotonicity Lemma]
\label{th:monotonicity}
Let $U_1$ and $U_2$ be two TA templates with conjunctive guards.
Let $\Phi(1_l)$ be an IMTL formula, with $l \in \{1, 2\}$.
Then for any $n \in \nat$
such that $n \geq 1$ we have:
$$
\begin{array}{cl}
(i)  ~ & ~ \!\!(U_1,U_2)^{(1,n)}   \models Q \Phi(1_2) \then 
	       (U_1,U_2)^{(1,n+1)} \models Q\Phi(1_2) \\	
(ii) ~ & ~ \!\!(U_1,U_2)^{(1,n)}   \models Q \Phi(1_1) \then 
	       (U_1,U_2)^{(1,n+1)} \models Q\Phi(1_1) \\
\end{array}
$$
where $Q \in \{ E, E_\textit{inf}, E_\textit{fin} \}$.
\end{theorem}

\ifextended
\begin{proof} 
\ \ \\
(i) The first part of the theorem states that there exists a 
\emph{s-path} $\rho$ of $(U_1,U_2)^{(1,n)}$ such that $\rho \models \Phi(1_2)$.
For each \emph{timed-computation} $x = (\mathfrak{c}_0,t_0)(\mathfrak{c}_1,t_1),\dots$, $x \in \mathit{tcomp}(\rho)$ 
it is possible to build a timed-word $y = (\mathfrak{y}_0, t_0)(\mathfrak{y}_1, t_1), \dots$ of $(U_1,U_2)^{(1,n+1)}$ 
such that at every step $v$,
$$
y_v = ( x_v(1,1) , x_v(2,1) , \dots, x_v(2,n), \hat{\mathfrak{s}}_v(2,n+1) )
$$
that is, the $n+1$-th instance of $U_2$ always remains (i.e. stutters) in its initial state $\hat{s}_2$, 
the rest of the automaton instances behave as in $x$. 
Since, in this case, each conjunctive guard has the following form
\begin{equation}
\label{guard}
	(\hat{s}_1^1 \vee p_1^1 \vee \dots \vee q_1^1) ~ \wedge ~ 
	\bigwedge_{\substack{1 < i \leq n}} (\hat{s}_2^{i} \vee p_2^{i} \vee \dots \vee q_2^{i})
\end{equation}
(i.e., it includes the initial state for each instance)
adding the $n+1$-th instance stuttering in its initial state does not change the truth value of such guards, then the conclusion holds. \\
(ii) The second part of the theorem follows from a very similar argument. \\
\qed
\end{proof}
\fi
Intuitively, from any time 
computation $x$ one can build a new time computation $y$ where each instance 
behaves as in $x$, except for a new instance of $U_2$ that halts in its initial 
state (remember that by definition the initial states don't falsify any 
conjunctive guard). 

\begin{theorem}[\bf Conjunctive Bounding Lemma]
\label{th:bounding}
Let $U_1$ and $U_2$ be two TA templates with conjunctive guards.
Let $\Phi(1_l)$ be an IMTL formula, with $l \in \{1, 2\}$.
Then for any $n \in \nat$
such that $n \geq 1$ we have:
$$
\begin{array}{cl}
(i) & \forall n \geq c_2 . (U_1,U_2)^{(1,n)} \models Q \Phi(1_2) ~ \then ~ (U_1,U_2)^{(1,c_2)} \models Q \Phi(1_2)\\
(ii) & \forall n \geq c_1 .(U_1,U_2)^{(1,n)} \models Q \Phi(1_1) ~ \then ~ (U_1,U_2)^{(1,c_1)} \models Q \Phi(1_1) \\
\end{array}
$$
where $Q \in \{ E, E_\textit{inf}, E_\text{fin} \}$ and:
\begin{itemize}
\item $c_1 = 1$ and $c_2 = 2$, when $Q = E_\textit{inf}$; 
\item $c_1 = c_2 = 1$, when $Q = E_\textit{fin}$; 
\item $c_1 = 2|U_2|$ and $c_2 = 2|U_2| + 1$, when $Q = E$.
\end{itemize}
\end{theorem}

\ifextended
\begin{proof}
\ \ \\
(i) ($\Rightarrow$) 

The first part of the theorem states that there exists a 
\emph{s-path} $\rho$ of $(U_1,U_2)^{(1,n)}$ such that $\rho \models \Phi(1_2)$. 
Let $x = (\mathfrak{c}_0,t_0)(\mathfrak{c}_1,t_1),\dots$ be a \emph{timed-computation} such that 
$x \in \mathit{tcomp}(\rho)$.
Then, it is possible to distinguish three distinct cases: 
$x$ is an infinite computation, a deadlocked computation, or a finite computation.

Let us suppose that $x$ is an \emph{infinite computation}.
This means that $x(1,1)$ is an infinite local computation or 
there exists $i \leq n$ such that $x(2,i)$ is an infinite local computation
(they are not mutually exclusive).
Then, it is possible to build a timed computation  
$y = (\mathfrak{y}_0, t_0)(\mathfrak{y}_1, t_1) \dots$ of $(U_1,U_2)^{(1,c_2)}$ 
as follows:
\begin{eqnarray}
&& \hspace*{-.7cm} 
   y(1,1) = x(1,1)                                                    \label{y11}     \\
&& \hspace*{-.7cm} 
   y(2,1) = x(2,1)                                                    \label{y21}     \\
&& \hspace*{-.7cm} 
   y(2,2) = \left\{
       \begin{array}{ll}
           x(2,i)                    & \textrm{if $x(1,1)$ and $x(2,1)$ are finite,} \\
                                     & \textrm{and $x(2,i)$ with $i>1$ is infinite,} \\
           \hat{\mathfrak{s}}(2,2)   & \textrm{otherwise}
       \end{array}
   \right.                                                            \label{y22}     \\
&& \hspace*{-.7cm} 
   \forall j \in [ 3 , c_2 ] . y(2,j) = \hat{\mathfrak{s}}(2,j)       \label{y2j}     
\end{eqnarray}
Rule \ref{y2j} can be applied only when $c_2=2|U_2| + 1$.
It should be noticed that $y$ preserves the local timed computation of $U_2^1$ and 
is a timed computation of $(U_1,U_2)^{(1,c_2)}$ (consider that either $c_2=2$ or $c_2=2|U_2| + 1$).
Indeed, for each $v$, five cases can occur:
\begin{description}
\item[I.1.] 
     $\textit{state}(x_{v+1})=\textit{state}(x_v)$. \\
     Therefore, $\textit{state}(y_{v+1})=\textit{state}(y_v)$ by construction. 
\item[I.2.] 
     $\textit{state}(x_v(1,1))\xrightarrow{g,r,\gamma}\textit{state}(x_{v+1}(1,1))$. \\
     This means that $\textit{state}(\mathfrak{c}_v)\models\gamma$.
     As $\gamma$ has the form reported in Equation (\ref{guard}), this implies $\textit{state}(\mathfrak{y}_v)\models\gamma$ 
     (instances $U_2^i$ eventually stuttering in their initial states do not prevent the progress of instance $U_1^1$).
     Furthermore, Case I.2 means that $\textit{clock}(\mathfrak{c}_v(1,1))\models g$ and, thus, 
     $\textit{clock}(\mathfrak{y}_v(1,1))\models g$ by construction \linebreak (Rule \ref{y11}). \\
     Therefore, $\textit{state}(y_v(1,1))\xrightarrow{g,r,\gamma}\textit{state}(y_{v+1}(1,1))$.
\item[I.3.] 
     $\textit{state}(x_v(2,1))\xrightarrow{g,r,\gamma}\textit{state}(x_{v+1}(2,1))$. \\
     Similarly to Case I.2, it is possible to prove that 
     $\textit{state}(y_v(2,1))\xrightarrow{g,r,\gamma}\textit{state}(y_{v+1}(2,1))$.
\item[I.4.] 
     $\textit{state}(x_v(2,i))\xrightarrow{g,r,\gamma}\textit{state}(x_{v+1}(2,i))$ and \\ $y(2,2)=x(2,i)$.  
     Similarly to Case I.2, it is possible to prove that 
     $\textit{state}(y_v(2,2))\xrightarrow{g,r,\gamma}\textit{state}(y_{v+1}(2,2))$.
\item[I.5.] 
     $\textit{state}(x_v(2,i))\xrightarrow{g,r,\gamma}\textit{state}(x_{v+1}(2,i))$ and \\ $y(2,2)\neq x(2,i)$. \\
     Therefore, 
     $\textit{state}(y_v)=\textit{state}(y_{v+1})$ by construction.
\end{description}

Let us suppose that $x$ is an \emph{deadlocked computation}, i.e. there exists $v$ such that, 
for each $l, i$, for each $x_v(l,i) \xrightarrow{g,r,\gamma} q \in \tau_l^i$, $\textit{state}(x_v) \not\models \gamma$, 
where $\gamma$ has the form reported in Equation (\ref{guard}).
This means that there exists $h, j$ such that $x_v(h,j) \not\in \{\hat{s}_h^{j} \vee p_h^{j} \vee \dots \vee q_h^{j}\}$ of $\gamma$.

Then, it is possible to build a timed computation 
$y = (\mathfrak{y}_0, t_0)(\mathfrak{y}_1, t_1) \dots$ of $(U_1,U_2)^{(1,2|U_2| + 1)}$ 
that preserves the deadlock as follows:
\begin{eqnarray}
&& \hspace*{-.7cm} 
   y(1,1) = x(1,1)                                                    \label{y11b}           \\
&& \hspace*{-.7cm} 
   y(2,1) = x(2,1)                                                    \label{y21b}           \\
&& \hspace*{-.7cm} 
   \forall k \in [ 2 , 2|U_2| + 1 ] .                                 \nonumber              \\
&& \hspace*{-.2cm} 
   y(2,k) = \left\{
       \begin{array}{ll}
           x(2,j)                    & \textrm{\!\!\!if there exists a distinct $y_v(l,i)$}  \\
                                     & \textrm{\!\!\!that is deadlocked by $x_v(2,j)$,}      \\
           \hat{\mathfrak{s}}(2,k)   & \textrm{\!\!\!otherwise}
       \end{array}
   \right.                                                            \label{y2jb}     
\end{eqnarray}
It should be noticed that, at most, there are $|U_2|$ distinct states that deadlock a local computation.
Nevertheless, when a given local state $x_v(2,i)=p$ is deadlocked by a given local state $x_v(2,j)=p$ 
(i.e. $x_v(2,i)=x_v(2,j)$), 
then the reverse holds as well, i.e., $x_v(2,j)$ is deadlocked by $x_v(2,i)$.
This means that in the worst case, the construction of $y$ requires at most $2|U_2|$ copies of local computations of $x$
(as guards are not reflexive).
Finally, $y$ needs at least an idle local computation of $U_2$ stuttering in its initial state to assure that each guard 
used to build $x$ can be fired in $y$ as well.
This justifies the choice of $c_2=2|U_2| + 1$.
At this point, it is possible to prove, similarly to the previous case, that 
$y$ preserves the local timed computation of $U_2^1$ and 
is a timed computation of $(U_1,U_2)^{(1,c_2)}$.
Indeed, for each $v$, five cases can occur:
\begin{description}
\item[D.1.] 
     $\textit{state}(x_{v+1})=\textit{state}(x_v)$. \\
     Therefore, $\textit{state}(y_{v+1})=\textit{state}(y_v)$ by construction. 
\item[D.2.] 
     $\textit{state}(x_v(1,1))\xrightarrow{g,r,\gamma}\textit{state}(x_{v+1}(1,1))$. \\
     This means that $\textit{state}(\mathfrak{c}_v)\models\gamma$.
     As $\gamma$ has the form reported in Equation (\ref{guard}), this implies $\textit{state}(\mathfrak{y}_v)\models\gamma$ 
     (instances $U_2^i$ eventually stuttering in their initial states do not prevent the progress of instance $U_1^1$).
     Furthermore, Case I.2 means that $\textit{clock}(\mathfrak{c}_v(1,1))\models g$ and, thus, 
     $\textit{clock}(\mathfrak{y}_v(1,1))\models g$ by construction \linebreak (Rule \ref{y11}). \\
     Therefore, $\textit{state}(y_v(1,1))\xrightarrow{g,r,\gamma}\textit{state}(y_{v+1}(1,1))$.
\item[D.3.] 
     $\textit{state}(x_v(2,1))\xrightarrow{g,r,\gamma}\textit{state}(x_{v+1}(2,1))$. \\
     Similarly to Case I.2, it is possible to prove that 
     $\textit{state}(y_v(2,1))\xrightarrow{g,r,\gamma}\textit{state}(y_{v+1}(2,1))$.
\item[D.4.] 
     $\textit{state}(x_v(2,j))\xrightarrow{g,r,\gamma}\textit{state}(x_{v+1}(2,j))$ and \\ $y(2,k)=x(2,j)$.  
     Similarly to Case I.2, it is possible to prove that 
     $\textit{state}(y_v(2,k))\xrightarrow{g,r,\gamma}\textit{state}(y_{v+1}(2,k))$.
\end{description}

Let us suppose that $x$ is a \emph{finite computation}.
This means that each local computation is finite. 
Then, it is possible to build a timed computation  
$y = (\mathfrak{y}_0, t_0)(\mathfrak{y}_1, t_1) \dots$ of $(U_1,U_2)^{(1,c_2)}$ 
as follows:
\begin{eqnarray}
&& \hspace*{-.7cm} 
   y(1,1) = x(1,1)                                                    \label{y11t}     \\
&& \hspace*{-.7cm} 
   y(2,1) = x(2,1)                                                    \label{y21t}     \\
&& \hspace*{-.7cm} 
   \forall j \in [ 2 , c_2 ] . y(2,j) = \hat{\mathfrak{s}}(2,j)       \label{y2jt}     
\end{eqnarray}
Rule \ref{y2jt} can be applied only when $c_2=2|U_2| + 1$.
It should be noticed that $y$ preserves the local timed computation of $U_2^1$ and 
is a timed computation of $(U_1,U_2)^{(1,1)}$ (consider that either $c_2=1$ or $c_2=2|U_2| + 1$).
Indeed, for each $v$, four cases can occur:
\begin{description}
\item[F.1.] 
     $\textit{state}(x_{v+1})=\textit{state}(x_v)$. \\
     Therefore, $\textit{state}(y_{v+1})=\textit{state}(y_v)$ by construction. 
\item[F.2.] 
     $\textit{state}(x_v(1,1))\xrightarrow{g,r,\gamma}\textit{state}(x_{v+1}(1,1))$. \\
     This means that $\textit{state}(\mathfrak{c}_v)\models\gamma$.
     As $\gamma$ has the form reported in Equation (\ref{guard}), this implies $\textit{state}(\mathfrak{y}_v)\models\gamma$ 
     (instances $U_2^i$ eventually stuttering in their initial states do not prevent the progress of instance $U_1^1$).
     Furthermore, Case F.2 means that $\textit{clock}(\mathfrak{c}_v(1,1))\models g$ and, thus, 
     $\textit{clock}(\mathfrak{y}_v(1,1))\models g$ by construction \linebreak (Rule \ref{y11t}). \\
     Therefore, $\textit{state}(y_v(1,1))\xrightarrow{g,r,\gamma}\textit{state}(y_{v+1}(1,1))$.
\item[F.3.] 
     $\textit{state}(x_v(2,1))\xrightarrow{g,r,\gamma}\textit{state}(x_{v+1}(2,1))$. \\
     Similarly to Case F.2, it is possible to prove that 
     $\textit{state}(y_v(2,1))\xrightarrow{g,r,\gamma}\textit{state}(y_{v+1}(2,1))$.
\item[F.4.] 
     $\textit{state}(x_v(2,i))\xrightarrow{g,r,\gamma}\textit{state}(x_{v+1}(2,i))$ and $i\geq 2$. \\
     Therefore, 
     $\textit{state}(y_v)=\textit{state}(y_{v+1})$ by construction.
\end{description}

In conclusion, it has been built a new Krypke Structure such that 
for each $x \in \textit{tcomp}(\rho)$, with $\rho$ starting from $x_0$, there 
exists $y \in \textit{tcomp}(\rho^\prime)$ such that $\rho^\prime$ starts in $y_0$
and it preserves the local timed computation of $U_2^1$.

\ \ \\
(i)($\Leftarrow$) \\
The opposite direction can be easily proved by means of repeated applications of the 
Monotonicity Lemma \ref{th:monotonicity}.

\ \ \\
(ii) \\
This part can be proved by applying similar arguments. \\
\qed
\end{proof}

\fi

\begin{theorem}[\bf Truncation Lemma]
\label{th:truncation}
Let $U_1$ and $U_2$ be two TA templates with conjunctive guards.
Let $\Phi(1_l)$ be an IMTL formula, with $l \in \{1, 2\}$, then:
$$
\begin{array}{l}
\forall n_1, n_2 \geq 1 . (U_1,U_2)^{(n_1,n_2)} \models Q \Phi(1_2) \textit{ iff } 
(U_1,U_2)^{(n_1^\prime,n_2^\prime)} \models Q \Phi(1_2)
\end{array}
$$
where $Q \in \{ E, E_\textit{inf}, E_\text{inf} \}$, 
$n_1^\prime = \textit{min}(n_1 , c_1)$, $n_2^\prime = \textit{min}(n_2 , c_2)$, 
and:
\begin{itemize}
\item $c_1 = 1$ and $c_2 = 2$, when $Q = E_\textit{inf}$; 
\item $c_1 = c_2 = 1$, when $Q = E_\textit{fin}$; 
\item $c_1 = 2|U_2|$ and $c_2 = 2|U_2| + 1$, when $Q = E$.
\end{itemize}
\end{theorem}

\ifextended
\begin{proof} 
First of all, let us prove that \\
$(U_1,U_2)^{(n_1,n_2)} \models E\Phi(1_2)$ iff 
$(U_1,U_2)^{(n_1,n_2^\prime)} \models E\Phi(1_2)$. \\
If $n_2 \leq c_2$ it is straightforward. If $n_2 > c_2$, let us set 
$V_1 = U_1^{n_1}$ and $V_2 = U_2$. Then: \\
$(U_1,U_2)^{(n_1,n_2)} \models E\Phi(1_2)$ iff \\
$(V_1,V_2)^{(1,n_2)} \models E\Phi(1_2)$ iff \\
$(V_1,V_2)^{(1,n_2^\prime)} \models E\Phi(1_2)$ (by Bounding Lemma \ref{th:bounding}) iff \\
$(U_1,U_2)^{(n_1,n_2^\prime)} \models E\Phi(1_2)$. \\

Now, let us prove the lemma. 
If $n_1 \leq c_1$ it is straightforward.
If $n_1 > c_1$, let us set $V_1 = U_2^{n_2^\prime}$ and $V_2 = U_1$. 
Then: \\
$(U_1,U_2)^{(n_1,n_2^\prime)} \models E\Phi(1_2)$ iff \\
$(U_2,U_1)^{(n_2^\prime,n_1)} \models E\Phi(1_2)$ iff \\
$(V_1,V_2)^{(1,n_1)} \models E\Phi(1_2/1_1)$ (where index $1_2$ of $U_2$ has been substituted by index $1_1$ of $V_1$) iff \\
$(V_1,V_2)^{(1,n_1^\prime)} \models E\Phi(1_2/1_1)$ (by Bounding Lemma \ref{th:bounding}) iff \\
$(U_2,U_1)^{(n_2^\prime,n_1^\prime)} \models E\Phi(1_2)$ iff \\
$(U_1,U_2)^{(n_1^\prime,n_2^\prime)} \models E\Phi(1_2)$. \\
\qed
\end{proof}
\fi
Thanks to the Truncation Lemma and the duality between operators $A$ and $E$, 
the Conjunctive Cutoff Theorem can be easily proved.
The Cutoff Theorem together with the known decidability and complexity results of 
the model checking problems for various timed temporal logics \cite{Bouyer2009} 
justify the following decidability theorem.

\begin{theorem}[\bf Decidability Theorem]
\label{th:decidability}
\ \ \\
Let $(U_1, \dots, U_k)$ be a set of TA templates with conjunctive guards 
and let $\phi = \bigwedge_{i_{l_1},\dots,i_{l_h}}Q \Phi(i_{l_1},\dots,i_{l_h})$ 
where $Q \in \{ A, A_\textit{inf}, A_\textit{fin}, E, E_\textit{inf}, 
E_\textit{fin} \}$ and $\{ l_1, \dots, l_h \} \in 
[1,k]$. 
The parameterized model checking problem (under the continuous time semantics)
$$
\begin{array}{lc}
\forall (n_1,\dots,n_k) \succeq (1,\dots,1)      . (U_1, \dots, U_k)^{(n_1, \dots,n_k)} \models \phi 
\end{array}
$$
is:
\begin{itemize}
\item \textsc{undecidable} 
      when $\Phi$ is an IMTL formula;
\item \textsc{decidable} and \textsc{2-expspace}
      when $\Phi$ is an IMITL formula;
\item \textsc{decidable} and \textsc{expspace}
      when $\phi$ is a TCTL formula.
\end{itemize}
\end{theorem}
\begin{proof}
\ifextended
For the first two results, consider that the Cutoff Theorem reduces the 
parameterized model checking problem to an ordinary model checking problem.
The latter is undecidable for MTL and is decidable and \textsc{expspace}-Complete 
(i.e. $\textsc{DSPACE}(2^{O(n)})$) for MITL \cite{Bouyer2009}.
Since the model in the parameterized model checking problem has at most 
an exponential number of states (i.e. $n = O(k \cdot |U|^{|U|}$), where $|U| = 
max(|U_1|,\dots,|U_k|$)), and it is invoked at most 
$\Theta(|U|^k)$ times, then the following is an upper 
bound on the complexity of the parameterized model checking problem: 
$\Theta(|U|^k) \cdot 2^{O(2^{|U| log(|U|)})}$, thus the problem is \textsc{2-expspace}.
Concerning the third statement, an ordinary model checking problem for TCTL 
is decidable and \textsc{pspace}-Complete (i.e. $O(n^p)$, for some $p$) 
\cite{Bouyer2009}. The parameterized model checking problem invokes
$\Theta(|U|^k)$ times a TCTL model checking problem, whose state space is at most
exponential, thus the complexity of the former is
$\Theta(|U|^k) \cdot O(k^p \cdot 2^{p |U| log(|U|)})$, i.e. at most 
\textsc{expspace}.
\fi
\end{proof}

\section{Case Study}
\label{sec:case-study}

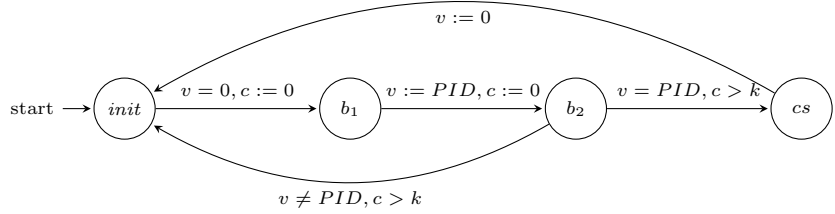
\begin{figure*}[t]
\center
\begin{tikzpicture}[->,>=stealth,shorten >=1pt,auto,node distance=3cm,font=\scriptsize]


    \node[initial,state] (qi)  {\emph{init}};
    \node[state]         (b1)  [right of=qi] {$b_1$};
    \node[state]         (b2)  [right of=b1] {$b_2$};
    \node[state]         (cs)  [right of=b2] {$cs$};

    \path (qi)  edge node {$v = 0, c := 0$} (b1);
    \path (b1)  edge node {$v := PID, c:= 0$} (b2);
    \path (b2)  edge node {$v = PID, c > k$} (cs)
                edge [bend left] node {$v \neq PID, c > k$} (qi);
    \path (cs)  edge [bend right] node {$v := 0$} (qi);

\end{tikzpicture}
\caption{\label{fig:proc_fischer} Process in Fischer's protocol as a Timed Automaton with integer variables}
\end{figure*}

We use the Fischer's protocol for mutual exclusion to show how to model-check a 
parameterized and timed systems. The protocol uses a single timed automaton 
template, instantiated an arbitrary number of time. Fig. \ref{fig:proc_fischer} 
depicts such template, where $\textit{inv}(b_1) = (c \leq k)$ \cite{Carioni2010}.
In Fischer's protocol every process (a) reads and writes a PID from and into a 
shared variable, and (b) waits a constant amount of time between when it asks to 
enter the critical section, and when it actually does so.
The Fischer's protocol cannot be directly modeled in our framework because of the
shared variable. We will first abstract the variable into a finite state system 
with conjunctive guards, and subsequently we will present the results of our 
verification. \\
\textbf{Abstracting Process Identifier.}
A variable can be modeled naively as an automaton with the structure of 
a completely connected graph, whose vertices denote possible assigned values
(let us call $V$ such model).
The state space can thus be infinite or finite, but even in the latter case it is
usually too big and makes the verification task unfeasible.

\begin{figure}[t]
\begin{minipage}[t][3cm][t]{0.45\textwidth}
\centering
\begin{tikzpicture}[->,node distance=2cm]

  \node[initial,state] (s0)  {$s_0$};
  \node[state]         (s1)  [above of=s0] {$s_1$};
  \node[state]         (s2)  [above right of=s0] {$s_2$};
  \node[state]         (s3)  [right of=s0,draw=none] {$\cdots$};

    \path (s0)  edge (s1)
                edge (s2)
                edge (s3);
    \path (s1)  edge (s2)
                edge [bend left] (s0)
                edge [bend left=75] (s3);
    \path (s2)  edge [bend left] (s0)
                edge [bend right] (s1)
                edge (s3);
    \path (s3)  edge [bend left] (s0)
                edge [bend right] (s2)
                edge [bend right=100] (s1);

\end{tikzpicture}
\caption{V: a shared variable}
\end{minipage} ~
\begin{minipage}[t][3cm][t]{0.45\textwidth}
\centering
\begin{tikzpicture}[->,node distance=2cm]

  \node[initial,state]         (np)  {\emph{diffpid}};
  \node[state]         (mp)  [right of=np]   {\emph{mypid}};

    \path (np)  edge [bend left] (mp)
                edge [loop above] (np);
    \path (mp)  edge [loop above] (mp)
                edge [bend left] (np);

\end{tikzpicture}
\caption{W: a process-centric view of a shared PID variable}

\end{minipage}
\end{figure}

An abstract shared variable for PIDs can be defined, under the assumptions:
\begin{itemize}
\item the variable only stores PID values;
\item the variable is shared among all processes;
\item every PID value overwrites the previous values of the variable itself;
\item every process can compare the variable value only with its own
    PID value.
\end{itemize}
\vspace{-0.2cm}
As in a predicate abstraction, we replace the shared variable with its
\emph{process-centric view}. The latter has only two relevant states: it is 
either the same PID as the process, or it stores a different one.
We use $W$ to denote such process.
Every process $P$ is in a one-to-one relation with its own view of the variable. 
We introduce a process template $P^\prime = P \times W$ that results from the
synchronous product of the $P$ and $W$. We could then model check a system 
${P^\prime}^{(n)}$. Doing this, we would probably obtain many spurious 
counter-examples, since two processes could have their copy of $W$ in state 
\texttt{*\_\_Mypid}. Since no variable can store multiple values, this is 
impossible. Conjunctive guards, though, allow to constraint
the system in such a way that no two processes can be in a state of the
\texttt{*\_\_Mypid} group. This solution rules out the undesired spurious behaviors,
and is very convenient since it can be applied whenever an algorithm uses a shared 
variable. We thus define $P^{\prime\prime}$ to be the refined version of $P^\prime$
represented in Fig. \ref{fig:fischer_uppaal} using the Uppaal notation.
It is possible to show that the abstract system simulates the concrete system, 
namely $(P \times V)^{(1,n)} \preceq (P \times W)^{(1,n)}$, for any positive $n$.

\begin{figure}[t]

\begin{minipage}[t][2.5cm][t]{0.45\linewidth}
\centering\includegraphics[width=\linewidth]{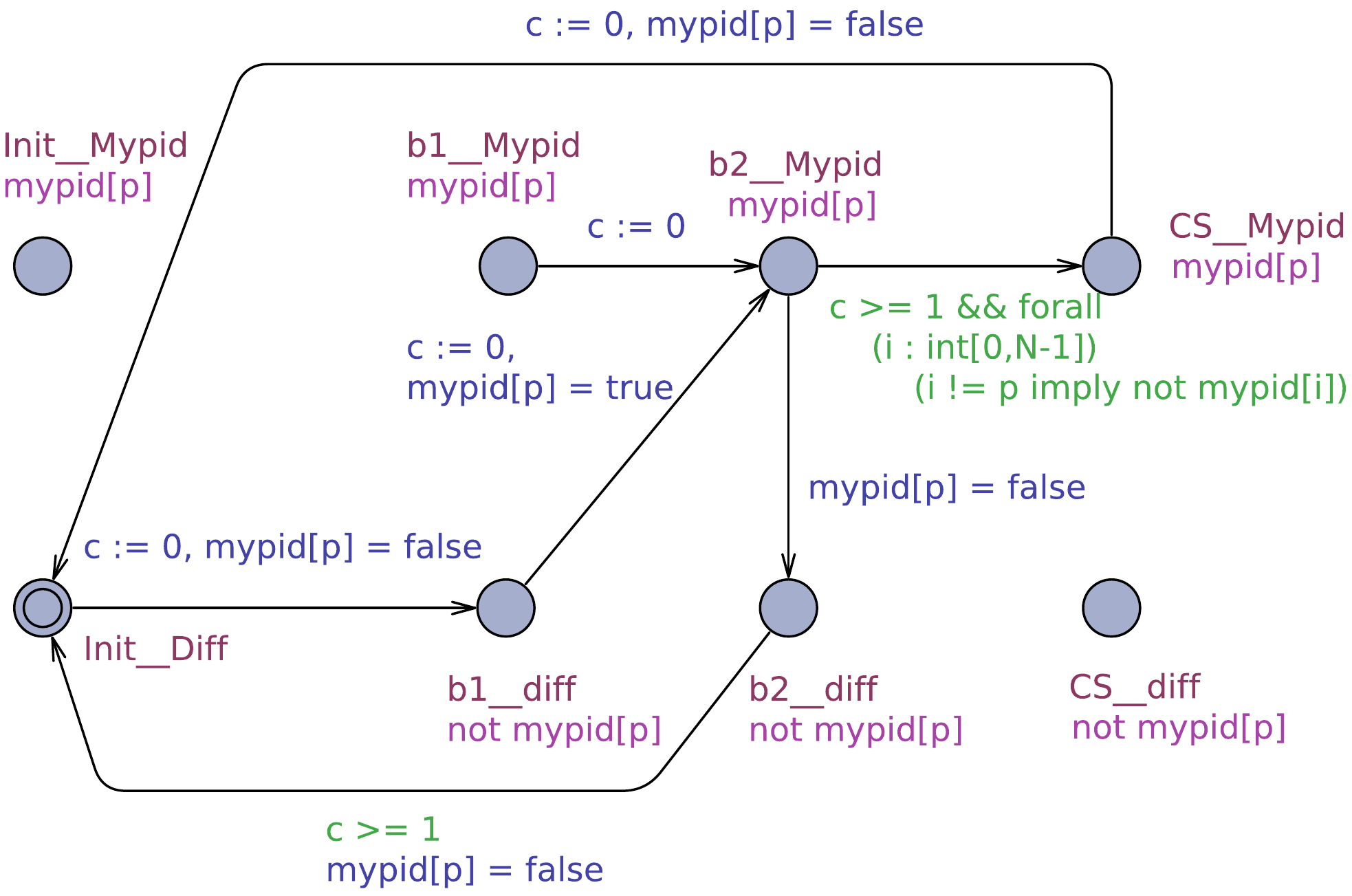}
\caption{\label{fig:fischer_uppaal}$P^{\prime\prime} = (P \times W) + CG$ template}
\end{minipage}
~
\begin{minipage}[t][2.5cm][t]{0.5\linewidth}
\centering\includegraphics[width=\linewidth]{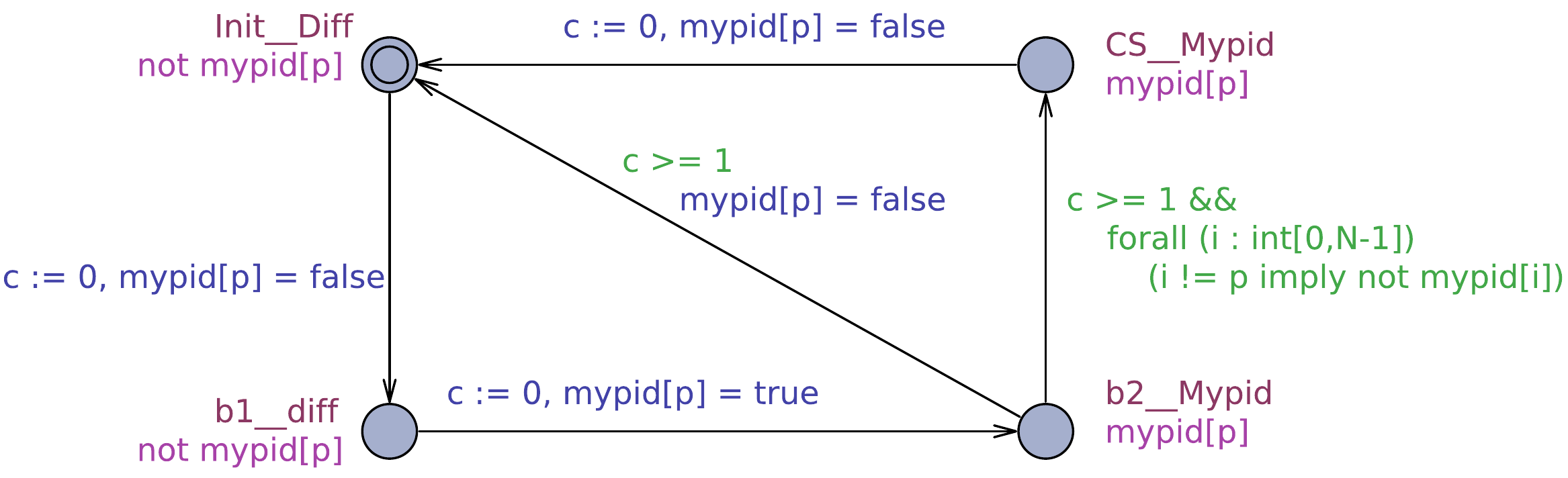}
\caption{\label{fig:fischer_uppaal_reduced}Reduced $P^{\prime\prime}$}
\end{minipage}

\end{figure}

Fig. \ref{fig:fischer_uppaal} depicts template $P^{\prime\prime}$. Some of the 
eight states
resulting from the product are not reached by any transition, and can thus be
removed from the model, implying a smaller cutoff. The model
manipulation up to this point can be completely automatized. We notice that it is 
safe to remove state \verb+b2__diff+ and connect directly 
state \verb+b2__Mypid+ with \verb+Init__Diff+, obtaining the reduced 
system in Fig. \ref{fig:fischer_uppaal_reduced}.
Finally, let us remark that variable \verb+mypid+ in Figg. \ref{fig:fischer_uppaal} and \ref{fig:fischer_uppaal_reduced} is added to overcome Uppaal 
syntax limitations that cannot refer directly to process states in guards and 
specifications. The reduced system has $4$ states, and thus the cutoff is $9$. \\
\\
\textbf{Verification Results.} 
Below are the formulae that have been model checked, together with the required
time and memory. \footnote{The experiments were run on an Intel Core2 Duo CPU 
T5870 @ 2.0 Ghz with 4GB RAM, OS Linux 3.13-1-amd64} \\
\begin{tabular}{llccc}
~ & \textit{Formula} & \textit{Outcome} & \textit{Time (s)} & \textit{Mem. (MB)} \\
(1) & $\bigwedge_i E \future_{\geq 0} (\textit{CS\_mypid}(i))$ & \textit{true} & 0.01 & 155.2 \\
(2) & $\bigwedge_{i \neq j} A \glob_{\geq 0} ! (\textit{CS\_mypid}(i) \wedge \textit{CS\_mypid}(j))$ & \textit{true} & 30.1 & 155.2 \\
(3) & $\bigwedge_i A \future_{\geq 0} (\textit{CS\_mypid}(i))$ & \textit{false} & 0.59 & 155.2 \\
\end{tabular} \vspace{0.2cm} \\
Formula (1) checks that a process can enter its critical section, while (2) 
checks the actual mutual exclusion property. Finally (3) states that a process will always 
be able to enter its critical section.
It is well known that while the Fischer's protocol ensures the
mutual exclusion property (i.e. formulae (1) and (2)), it also suffers from the 
problem of processes to possibly starve (i.e. formula (3)).

\section{Conclusions}

In this work we presented the combined study of timed and parameterized systems. 
We proved that a cutoff exists for PNTA with conjunctive guards and a subset of 
ITCTL$^\star$ formulae. Moreover, the cutoff value is equal to the value  
computed in Emerson and Kahlon's work for untimed systems\cite{Emerson2000}.
This proves that the parameterized model checking problem 
is decidable for networks of timed automata with disjunctive guards, for a 
suitable logic. We remark that for timed systems, applying 
Thm. \ref{the:conj-cutoff} one obtains a considerably smaller cutoff than 
applying the  (untimed) Emerson and Kahlon's cutoff theorem after reducing the 
original timed system to a finite state system by means of the traditional region 
or zone abstractions.

Finally, we used the Fischer's protocol for mutual exclusion as a benchmark for 
showing how to apply the cutoff theorem. We claim that the use of conjunctive 
guards is convenient
for verifying systems based of shared variables, since they naturally express the
constraint that a variable can store only one value at any time.
As a follow-up of this work, we aim at two main goals: (a) finding more algorithms
for real-time and distributed systems that can be model checked using our 
framework, and (b) extending the Emerson and Kahlon cutoff theorem also to PNTA 
with Disjunctive Guards.

\bibliographystyle{plain}
\bibliography{biblio}

\end{document}